\documentclass[11pt,a4paper]{article}

\usepackage{amsmath,amssymb,amsthm}
\usepackage{enumitem}
\usepackage{hyperref}
\usepackage{geometry}
\usepackage{algorithm}
\usepackage{algpseudocode}
\usepackage{booktabs}
\usepackage{array}
\usepackage{orcidlink}     
\usepackage{url}

\newtheorem{assumption}{Assumption}

\usepackage{draftwatermark}
\SetWatermarkText{\textsc{Draft}}
\SetWatermarkScale{3}
\SetWatermarkColor[gray]{0.9}

\hypersetup{
  colorlinks=true,
  linkcolor=black,
  citecolor=black,
  urlcolor=black
}

\newtheorem{theorem}{Theorem}[section]
\newtheorem{proposition}[theorem]{Proposition}
\newtheorem{corollary}[theorem]{Corollary}

\theoremstyle{definition}
\newtheorem{definition}[theorem]{Definition}
\newtheorem{example}[theorem]{Example}
\theoremstyle{remark}
\newtheorem{remark}[theorem]{Remark}
\newtheorem{openprob}[theorem]{Open Problem}

\newcommand{\keyspace}{\mathcal{K}}
\newcommand{\Hent}{H}
\newcommand{\floor}[1]{\lfloor #1 \rfloor}
\newcommand{\ceil}[1]{\lceil #1 \rceil}

\begin{document}

\title{From Bits to Mixed-Radix Keys:\\
       Horner Decomposition, Uniform Sampling, and\\
       the Information-Theoretic QKD Interface of the MR-OTP}

\author{Fabio F.G. Buono\thanks{Independent Researcher}\,\orcidlink{0009-0004-9199-2793}}
\date{\today}
\maketitle

\maketitle

\maketitle

\begin{center}
    {\large \textit{Preprint}}
\end{center}

\begin{abstract}
The Mixed-Radix One-Time Pad (MR-OTP)~\cite{buono2026a} generalizes the classical one-time pad to digit spaces with non-uniform bases, preserving Shannon perfect secrecy while enabling native encoding of heterogeneous
data alphabets without conversion to binary. A practical deployment of the MR-OTP requires a concrete, bias-free procedure for converting raw binary entropy, as produced by a Quantum Key Distribution (QKD) source, into a key tuple that is uniform over the
mixed-radix key space. This paper identifies Horner's method~\cite{horner1819} and its inverse as the canonical algebraic correspondence for this conversion, grounded in the structural equivalence between the mixed-radix positional
weights of~\cite{buono2026a} and the nested form of polynomial evaluation. We prove that naive modular reduction introduces a statistical bias that violates the uniformity hypothesis of the MR-OTP perfect secrecy
theorem (Proposition~\ref{prop:bias}), and that rejection sampling combined with inverse Horner decomposition restores uniform distribution at an expected bit cost of at most $2\ceil{\log_2 P}$
(Theorem~\ref{thm:uniformity}, Proposition~\ref{prop:cost}). We prove end-to-end perfect secrecy for the full QKD pipeline in one session (Theorem~\ref{thm:e2e}) and across $N$ sessions under a
partition key-rolling scheme (Theorem~\ref{thm:multisession}). We prove that the information-theoretic security guarantee is preserved under any future algorithmic advance on the Base Recovery Problem,
including a polynomial-time solution (Theorem~\ref{thm:invariance}), extend this to message distributions with arbitrary restricted support (Corollary~\ref{cor:restricted-support}), and establish an unconditional
lower bound on the query complexity of base recovery in the known-plaintext setting (Theorem~\ref{thm:query-lower-bound}). We present a complete adversary taxonomy separating testing-based,
algebraic, unbounded, and ciphertext-only adversaries, proving that base recovery is information-theoretically insoluble in the ciphertext-only setting (Corollary~\ref{cor:co-insoluble}) and establishing that the
information-theoretic guarantee on messages covers all classes unconditionally while the query lower bound and computational hardness cover disjoint subsets of the adversary space for the problem of
identifying the base sequence from known-plaintext data. We identify the logical structure of the invariance result as an instance of the Syntactic Invariance Principle of~\cite{buono2026b}.
We also quantify representational efficiency over binary OTP for natural alphabets (Proposition~\ref{prop:efficiency}), develop a batched extraction algorithm (Theorem~\ref{thm:batched}), and characterize
security degradation for non-ideal QKD sources (Theorem~\ref{thm:epsilon-security}). Five open problems identify the main directions for further work, with the product-support case of the optimal base selection problem resolved by Corollary~\ref{cor:restricted-support}.
\end{abstract}

\tableofcontents\bigskip

\paragraph{Notation.}
Throughout the paper, $\log_2$ denotes the base-$2$ logarithm, $\floor{\cdot}$ the floor function, and $\ceil{\cdot}$ the ceiling
function. The symbol $\mathbb{Z}$ denotes the integers and $\mathbb{Z}/m\mathbb{Z}$ the integers modulo $m$.
Given a base sequence $B=(b_1,\dots,b_L)$ with each $b_i\ge 2$, we write $P = \prod_{i=1}^L b_i$ for the product of all bases and
$\keyspace = \prod_{i=1}^L \{0,\dots,b_i-1\}$ for the corresponding key space. The symbol $\Hent(\cdot)$ denotes Shannon entropy measured in bits,
$\mathcal{R}$ denotes a source of independent uniform random bits, and all probability statements are taken over the randomness of $\mathcal{R}$.
\bigskip

\section{Introduction}
\label{sec:intro}

The one-time pad, introduced by Vernam in 1917 and patented in 1919~\cite{vernam1926}, achieves perfect secrecy in the sense that an
intercepted ciphertext yields no information about the plaintext regardless of the computational resources available to the adversary. Shannon formalized this property in 1949~\cite{shannon1949}, proving that
it holds if and only if the key is drawn uniformly at random, is at least as long as the message, and is never reused. The preceding paper~\cite{buono2026a} extended this classical construction
to \emph{mixed-radix} digit spaces by allowing each position of the message to reside in its own finite cyclic group, defining the MR-OTP and proving that Shannon's perfect secrecy theorem carries over without modification.
The construction originates from the earlier work~\cite{buono2012}, which introduced the idea of encoding data in a secret mixed-radix representation
and observed that the numerical value of a sequence is hidden from any process that operates on the symbols alone.

\paragraph{The MR-OTP as a unifying framework.}
The first paper~\cite{buono2026a} establishes the MR-OTP as a unifying
algebraic framework for perfect-secrecy encryption over finite digit spaces.
By allowing each position $i$ of the message to carry its own base
$b_i \ge 2$, the construction subsumes two classical families as special
cases.
When all bases are set to $b_i = 2$, modular addition at each position
reduces to bitwise XOR and the MR-OTP coincides with the binary OTP
(Corollary~1 of~\cite{buono2026a}).
When all bases are set to a common constant $b$, the construction
recovers the one-time pad over an alphabet of size $b$, for example
$b=26$ for Latin text or $b=4$ for nucleotide sequences
(Corollary~2 of~\cite{buono2026a}).
The extension to non-constant base sequences is the contribution of
\cite{buono2026a}, where each symbol of the data is encoded in its native
algebraic structure without any intermediate binary conversion.
This unification preserves the information-theoretic constraint on
key length without modification.
As shown in Proposition~2 of~\cite{buono2026a}, the pad must carry entropy
at least equal to that of the message regardless of how the bases are
chosen or whether they are kept secret, because the bases modify the
representation of the data but leave the entropy that the key
must supply unchanged.
The MR-OTP therefore achieves Shannon's bound~\cite{shannon1949} with
equality.

The present paper addresses the next algorithmic question, namely how to
convert the output of a concrete binary entropy source into a key that
satisfies the uniform distribution required by the MR-OTP.
Quantum key distribution (QKD) systems such as BB84~\cite{bb84} are the
most natural source for this purpose, as they produce a stream of
independent, uniformly distributed classical bits whose distribution is
guaranteed by the structure of the protocol, with no assumption on the
computational resources of an adversary.
The MR-OTP, however, requires a key $K = (k_1,\dots,k_L)$ in which each
digit $k_i$ is uniform over $\{0,\dots,b_i-1\}$ and the digits are
mutually independent, a condition that a binary stream satisfies only when
every base $b_i$ is a power of $2$, a requirement that the data alphabets
of practical interest do not meet.

Reading $k = \ceil{\log_2 P}$ bits and reducing the resulting integer
modulo $P$ introduces a statistical distortion known as modulo
bias~\cite{lemire2019}, which arises whenever $P$ is not a power of $2$.
In most algorithmic settings this bias is negligible, but in the
information-theoretic setting any non-uniformity in the key distribution
violates the hypothesis of the MR-OTP perfect secrecy theorem~\cite{buono2026a}
and renders the system insecure regardless of the magnitude of the
distortion.

\paragraph{Main contributions.}
\begin{enumerate}[label=\arabic*.]
  \item We identify Horner's method~\cite{horner1819} and its inverse
        as the canonical algebraic duality underlying the conversion from
        integers to mixed-radix tuples
        (Section~\ref{sec:horner}).
  \item We prove that naive modular reduction of a $k$-bit integer
        introduces a bias that breaks perfect secrecy
        (Proposition~\ref{prop:bias}), and that a rejection-sampling
        step restores uniform distribution at an expected bit cost of at
        most $2\ceil{\log_2 P}$ (Theorem~\ref{thm:uniformity},
        Proposition~\ref{prop:cost}).
  \item We prove end-to-end information-theoretic security of the full
        QKD pipeline for a single session (Theorem~\ref{thm:e2e}) and
        for $N$ sessions under a partition key-rolling scheme
        (Theorem~\ref{thm:multisession}).
  \item We prove that the information-theoretic security guarantee is
        preserved regardless of any computational advance on the Base
        Recovery Problem (Theorem~\ref{thm:invariance}), and extend
        this to message distributions with arbitrary restricted support
        (Corollary~\ref{cor:restricted-support}).
  \item We prove that the Base Recovery Problem is information-theoretically
        insoluble in the ciphertext-only setting
        (Corollary~\ref{cor:co-insoluble}), establish an unconditional
        lower bound on its query complexity in the known-plaintext setting
        (Theorem~\ref{thm:query-lower-bound}), separate the testing and
        algebraic adversary classes, and present the complete adversary
        taxonomy (Remark~\ref{rem:taxonomy}).
  \item We identify the logical structure of the invariance result
        as an instance of the Syntactic Invariance Principle
        of~\cite{buono2026b}, and use this framework to reformulate
        the dynamic key rolling open problem
        (Section~\ref{rem:sip}, Open Problem~\ref{op:dynamic}).
  \item We quantify representational efficiency over binary OTP
        (Section~\ref{sec:efficiency}), develop the batched key
        extraction algorithm (Section~\ref{sec:batched}), and
        characterize security degradation for non-ideal QKD sources
        (Theorem~\ref{thm:epsilon-security}).
\end{enumerate}

\paragraph{Organization.}
Section~\ref{sec:horner} reviews Horner's method, develops its inverse,
and establishes the forward--inverse duality.
Section~\ref{sec:sampling} proves the bias theorem, the uniformity theorem,
the cost bound, and states the complete algorithm.
Section~\ref{sec:qkd} proves end-to-end information-theoretic security
of the full QKD pipeline.
Section~\ref{sec:efficiency} quantifies representational efficiency over
binary OTP.
Section~\ref{sec:rolling} proves multi-session security and formulates
the dynamic key rolling open problem using the Syntactic Invariance
Principle framework.
Section~\ref{sec:open} states the five open problems.
Section~\ref{sec:hardness} analyzes the Base Recovery Problem in both
the ciphertext-only and known-plaintext settings, proves the
information-theoretic insolubility of the former
(Corollary~\ref{cor:co-insoluble}), establishes unconditional bounds in
the query model for the latter, and identifies the logical structure of
the temporal security result as an instance of the Syntactic Invariance
Principle of~\cite{buono2026b}.
Section~\ref{sec:batched} develops the batched key extraction algorithm.
Section~\ref{sec:epsilon} quantifies security degradation for non-ideal
QKD sources.

\section{Horner's Method and Its Inverse}
\label{sec:horner}

\subsection{Horner's method: polynomial evaluation in nested form}

Let $a_0, a_1, \dots, a_n$ be the coefficients of a polynomial and let
$x$ be an indeterminate.
Evaluating the polynomial in expanded form as
\[
  P(x) \;=\; a_n x^n + a_{n-1} x^{n-1} + \cdots + a_1 x + a_0
\]
requires $O(n^2)$ multiplications when each power $x^j$ is computed
independently by repeated squaring.
Horner's method~\cite{horner1819} eliminates this redundancy by rewriting
the same expression in nested form as
\[
  P(x) \;=\; a_0 + x\bigl(a_1 + x\bigl(a_2 + \cdots + x(a_{n-1} +
             x\,a_n)\cdots\bigr)\bigr),
\]
which requires only $n$ multiplications and $n$ additions and is optimal
by a classical lower-bound argument~\cite{knuth1997}.
The nested form is computed by the recurrence
\[
  b_n = a_n, \qquad b_j = a_j + x\,b_{j+1} \quad (j = n-1,\dots,0),
\]
with $P(x) = b_0$.

\begin{remark}
  Horner's method is valid over any ring, and in particular over
  $\mathbb{Z}$ and over $\mathbb{Z}/m\mathbb{Z}$ for any modulus $m$.
  It is this algebraic generality that makes the method applicable to
  mixed-radix numeral systems over finite cyclic groups.
\end{remark}

\subsection{The mixed-radix representation as a Horner evaluation}

Fix a length $L$ and a base sequence $B = (b_1,\dots,b_L)$ with each
$b_i \ge 2$.
Recall from~\cite{buono2026a} the positional weights
\[
  W_i \;=\; \prod_{j > i} b_j, \qquad W_L = 1.
\]
The integer value associated with a digit tuple $(k_1,\dots,k_L)$ is
\begin{equation}
  \label{eq:horner-forward}
  V \;=\; \sum_{i=1}^{L} k_i W_i
    \;=\; k_1(b_2 b_3 \cdots b_L) + k_2(b_3 b_4 \cdots b_L) + \cdots + k_L.
\end{equation}
Writing the same sum in nested form gives
\begin{equation}
  \label{eq:horner-nested}
  V \;=\; \bigl(\cdots\bigl((k_1 \cdot b_2 + k_2) \cdot b_3 + k_3\bigr)
          \cdots \bigr) \cdot b_L + k_L,
\end{equation}
which is a Horner evaluation of the digit sequence $(k_1,\dots,k_L)$
treated as coefficients, with the position-dependent multipliers
$b_2, b_3, \dots, b_L$ playing the role of the variable $x$ at each
nesting level.

\subsection{Inverse Horner decomposition}

\begin{definition}[Inverse Horner decomposition]
\label{def:inverse-horner}
Given a base sequence $B = (b_1,\dots,b_L)$ and an integer
$V \in \{0,\dots,\prod_i b_i - 1\}$, the \emph{inverse Horner
decomposition} recovers the digit tuple $(k_1,\dots,k_L)$ by the
following procedure.
\begin{enumerate}[label=\arabic*.]
  \item Set $V^{(L)} \leftarrow V$.
  \item For $i = L$ downto $1$, set
        $k_i \leftarrow V^{(i)} \bmod b_i$ and
        $V^{(i-1)} \leftarrow \floor{V^{(i)}/b_i}$.
  \item Return $(k_1,\dots,k_L)$.
\end{enumerate}
\end{definition}

\begin{proposition}[Correctness]
\label{prop:correctness}
For every $V \in \{0,\dots,\prod_i b_i - 1\}$, the inverse Horner
decomposition returns the unique tuple $(k_1,\dots,k_L)$ with
$0 \le k_i < b_i$ and $V = \sum_{i=1}^L k_i W_i$.
\end{proposition}

\begin{proof}
At step $i$, the division algorithm gives $V^{(i)} = k_i + b_i V^{(i-1)}$,
so $k_i = V^{(i)} \bmod b_i$ and $V^{(i-1)} = \floor{V^{(i)}/b_i}$.
Unrolling recovers the nested form~\eqref{eq:horner-nested}, hence
$V = \sum_i k_i W_i$. Uniqueness follows since the map is injective
on a set of cardinality $\prod_i b_i$.
\end{proof}

\begin{example}
\label{ex:decomposition}
$B = (7,13,5)$, $P = 455$, $V = 427$:
\[
  k_3 = 427 \bmod 5 = 2,\;\; V \leftarrow 85;\quad
  k_2 = 85  \bmod 13 = 7,\;\; V \leftarrow 6;\quad
  k_1 = 6   \bmod 7 = 6.
\]
$K = (6,7,2)$. Check: $6 \cdot 65 + 7 \cdot 5 + 2 = 390+35+2=427$. \checkmark
\end{example}

\subsection{The forward--inverse duality}

\begin{center}
\begin{tabular}{lll}
\toprule
\textbf{Direction} & \textbf{Name} & \textbf{Operation} \\
\midrule
$(k_1,\dots,k_L)\to V$ & Horner (forward) &
  $V = ((\cdots(k_1 b_2 + k_2)b_3+k_3)\cdots)b_L+k_L$ \\
$V \to (k_1,\dots,k_L)$ & Inverse Horner &
  $k_i = V^{(i)}\bmod b_i$,\; $V^{(i-1)}=\floor{V^{(i)}/b_i}$ \\
\bottomrule
\end{tabular}
\end{center}

The forward map composes a digit tuple into a single integer by
accumulating products in the manner of Horner evaluation, while the
inverse map recovers the digits from the integer by successive modular
reduction starting from the least significant position.
The two operations are mutual inverses, establishing a bijection between
$\keyspace$ and the integer interval $\{0,\dots,P-1\}$.

\section{Uniform Key Sampling: The Complete Algorithm}
\label{sec:sampling}

\subsection{Why naive conversion fails: modulo bias}

A binary entropy source such as a QKD device reads $k$ bits at a time
and delivers an integer $V$ that is uniform on $\{0,\dots,2^k-1\}$.
The inverse Horner decomposition requires $V$ to be uniform on the
smaller set $\{0,\dots,P-1\}$, and these two distributions agree only
when $P$ is itself a power of $2$.

\begin{proposition}[Modulo bias breaks perfect secrecy]
\label{prop:bias}
Let $k = \ceil{\log_2 P}$, $r = 2^k \bmod P > 0$.
If $V$ is uniform on $\{0,\dots,2^k-1\}$, then:
\begin{enumerate}[label=(\roman*)]
  \item Values $v \in \{0,\dots,r-1\}$ each have probability $2/2^k$;
        values $v \in \{r,\dots,P-1\}$ each have probability $1/2^k$.
  \item The digits $(k_1,\dots,k_L)$ from inverse Horner are
        \textbf{not} uniform over $\keyspace$.
  \item Perfect secrecy of the MR-OTP is violated.
\end{enumerate}
\end{proposition}

\begin{proof}
The map $v \mapsto v \bmod P$ sends $\{0,\dots,2^k-1\}$ to
$\{0,\dots,P-1\}$, with each residue $v'<r$ covered by
$\floor{2^k/P}+1$ pre-images and each $v'\ge r$ by $\floor{2^k/P}$.
Since $r>0$ these counts differ, so the distribution on $\{0,\dots,P-1\}$
is non-uniform. The inverse Horner map is a bijection
(Proposition~\ref{prop:correctness}), so the key distribution over
$\keyspace$ is non-uniform, violating the uniformity hypothesis of
Theorem~1 in~\cite{buono2026a}.
\end{proof}

\begin{example}
\label{ex:bias}
$B=(7,13,5)$, $P=455$, $k=9$, $2^9=512$, $r=57$.
The 57 keys corresponding to $V\in\{0,\dots,56\}$ each have probability
$2/512$; the remaining $455-57=398$ keys have probability $1/512$.
\end{example}

\subsection{Rejection sampling restores uniform distribution}

\begin{definition}[Rejection sampler $\mathcal{S}(B)$]
\label{def:rejection}
Given a base sequence $B$, let $P=\prod_i b_i$ and $k=\ceil{\log_2 P}$.
The rejection sampler $\mathcal{S}(B)$ proceeds through the following steps.
\begin{enumerate}[label=\arabic*.]
  \item Draw $k$ uniform bits from $\mathcal{R}$ and interpret them as
        an integer $V \in \{0,\dots,2^k-1\}$.
  \item If $V \ge P$, discard $V$ and return to step~1.
  \item Apply the inverse Horner decomposition to $V$ and return the
        resulting key $K$.
\end{enumerate}
\end{definition}

\begin{theorem}[Uniform sampling]
\label{thm:uniformity}
$\mathcal{S}(B)$ terminates with probability~$1$ and outputs $K$
uniform on $\keyspace$ with mutually independent digits.
\end{theorem}

\begin{proof}
\textit{Termination.} The acceptance probability per round is
$P/2^k \ge 1/2$ (since $2^k \le 2P$). Rounds are independent, so
termination is geometric with parameter $\ge 1/2$.

\textit{Uniformity upon acceptance.} For any $v \in \{0,\dots,P-1\}$:
\[
  \Pr[V=v \mid V < P]
  = \frac{1/2^k}{P/2^k} = \frac{1}{P}.
\]

\textit{Independence and uniformity of digits.} The inverse Horner map
is a bijection $\{0,\dots,P-1\} \to \keyspace$
(Proposition~\ref{prop:correctness}), so it carries the unique uniform
distribution on $\{0,\dots,P-1\}$ to the unique uniform distribution
on $\keyspace = \prod_i \{0,\dots,b_i-1\}$.

For independence, observe that $|\keyspace| = P = \prod_i b_i$.
For any fixed values $t_1 \in \{0,\dots,b_1-1\},\dots,
t_L \in \{0,\dots,b_L-1\}$, just one $v \in \{0,\dots,P-1\}$
maps to $(t_1,\dots,t_L)$, so
$\Pr[K = (t_1,\dots,t_L)] = 1/P = \prod_i (1/b_i)$.
This factorization of the joint distribution into a product of marginals
is mutual independence, with each $k_i$ uniform on $\{0,\dots,b_i-1\}$.
\end{proof}

\subsection{Expected bit cost}

\begin{proposition}[Bit cost bound]
\label{prop:cost}
The expected number of bits consumed by $\mathcal{S}(B)$ satisfies
\[
  \mathbb{E}[\text{bits}] = k \cdot \frac{2^k}{P} \;\le\; 2k
  = 2\ceil{\log_2 P}.
\]
\end{proposition}

\begin{proof}
Rounds are i.i.d., each using $k$ bits, with success probability
$p = P/2^k$. Expected rounds: $1/p = 2^k/P$. Total expected bits:
$k/p \le 2k$ since $2^k \le 2P$.
\end{proof}

\begin{remark}[Practical acceptance rates for natural alphabets]
\label{rem:acceptance}
\hfill
\begin{itemize}
  \item $B = (4)^L$ (nucleotides): $P = 2^{2L}$, $k = 2L$, no rejection
        ever needed. Expected rounds: $1$.
  \item $B = (26)^5$ (five Latin letters): $P = 11\,881\,376$, $k = 24$,
        acceptance $\approx 70.8\%$, expected $\approx 1.41$ rounds.
  \item $B = (26,10,8)$: $P = 2080$, $k = 12$
        (since $2^{11}=2048 < 2080 \le 2^{12}=4096$),
        acceptance $= 2080/4096 \approx 50.8\%$, expected $\approx 1.97$ rounds.
  \item $B = (20)^4$ (four amino acids): $P = 160\,000$, $k = 18$
        (since $2^{17}=131\,072 < 160\,000 \le 2^{18}=262\,144$),
        acceptance $\approx 61.0\%$, expected $\approx 1.64$ rounds.
\end{itemize}
In all cases the expected number of rounds is at most~$2$, confirming
that the overhead of rejection sampling is negligible in practice.
\end{remark}

\begin{remark}[Comparison with Lemire's method]
\label{rem:lemire}
Algorithm~\ref{alg:sample} uses the simplest form of rejection sampling,
which performs one integer division per accepted sample during the inverse
Horner decomposition.
Lemire~\cite{lemire2019} proposes a multiplication-based variant that
avoids integer division in most cases by computing the $2k$-bit product
$V \cdot P$ and accepting when the lower $k$ bits of the product exceed
the threshold $(-P) \bmod 2^k$, falling back to a division only on
rejection events.
Since the probability of rejection is $1 - P/2^k \le 1/2$, the expected
number of divisions per sample in Lemire's method is at most $1/2$,
compared to $1$ in Algorithm~\ref{alg:sample}.

From a computational standpoint, Lemire's method reduces the division
count by roughly a factor of $2$ on general-purpose processors where
integer divisions are slower than multiplications by a substantial
margin, and this advantage is largest when $P$ is close to $2^{k-1}$
and the rejection probability is highest.
From an entropic standpoint, both methods consume the same expected
number of bits per key, namely $k \cdot 2^k / P$, so the choice of
method leaves the rate of QKD-generated entropy consumption unchanged.

For implementations that pair the MR-OTP with a dedicated QKD hardware
source operating at rates in the megabit-per-second range, the throughput
of key tuples is bounded by the bit generation rate of the source, and
both methods deliver the same throughput.
On software platforms or embedded systems where division latency is a
concern, Lemire's variant is preferable and is a drop-in replacement for
the rejection step in Algorithm~\ref{alg:sample} without affecting any
of the security properties proved in Theorem~\ref{thm:uniformity} or
Theorem~\ref{thm:e2e}.

A further comparison arises from the work of Draper and
Saad~\cite{draperSaad2025} on entropy-optimal discrete sampling, which
shows that any unbiased sampler for a distribution over $P$ outcomes
must consume at least $\log_2 P = H(\mathrm{Uniform}(P))$ bits per
sample in expectation, with achievable algorithms reaching between
$H(\mathrm{Uniform}(P))$ and $H(\mathrm{Uniform}(P))+2$ bits per sample.
The expected cost of Algorithm~\ref{alg:sample} is $k \cdot 2^k/P$,
which can be written as $\log_2 P \cdot (2^k/P) \cdot (k/\log_2 P)$.
Since $2^k/P \le 2$ and $k/\log_2 P = \ceil{\log_2 P}/\log_2 P \le
1 + 1/\log_2 P$, the algorithm is within a small constant of the
entropy-optimal bound for all $P \ge 4$, and the gap vanishes as $P$
grows.
Achieving the tight $[H(P), H(P)+2]$ bound of Draper and Saad requires
arithmetic of considerably greater complexity, and for the cryptographic
application considered here the simpler algorithm is sufficient.
\end{remark}

\subsection{The complete algorithm}

\begin{algorithm}[H]
\caption{MR-OTP Key Extraction from a Binary Entropy Source}
\label{alg:sample}
\begin{algorithmic}[1]
\Require Base sequence $B=(b_1,\dots,b_L)$, $b_i\ge 2$;
         binary entropy source $\mathcal{R}$.
\Ensure  $K=(k_1,\dots,k_L)$ uniform on $\keyspace$.
\State $P \leftarrow \prod_{i=1}^{L} b_i$;\quad
       $k \leftarrow \ceil{\log_2 P}$
\Repeat
  \State Read $k$ bits from $\mathcal{R}$; interpret as integer $V$
\Until{$V < P$}
\For{$i = L$ \textbf{downto} $1$}
  \State $k_i \leftarrow V \bmod b_i$;\quad
         $V   \leftarrow \floor{V / b_i}$
\EndFor
\State \Return $(k_1,\dots,k_L)$
\end{algorithmic}
\end{algorithm}

\begin{corollary}
\label{cor:correct}
Algorithm~\ref{alg:sample} satisfies all preconditions of
Theorem~1 of~\cite{buono2026a}: it terminates almost sure, and outputs a key
uniform on $\keyspace$ with independent digits.
\end{corollary}

\subsection{Worked example with rejection step}

\begin{example}
\label{ex:full}
$B=(7,13,5)$, $P=455$, $k=9$.

\textit{Round~1.} Bits \texttt{110101011} $\Rightarrow V=427 < 455$.
\textbf{Accept.} Decompose: $K=(6,7,2)$ (see Example~\ref{ex:decomposition}).

\textit{Hypothetical rejection.} Bits \texttt{111111110}
$\Rightarrow V=510 \ge 455$. \textbf{Reject}; read 9 new bits.
\end{example}

\section{End-to-End Information-Theoretic Security with QKD}
\label{sec:qkd}

\subsection{The pipeline}

We now assemble the full pipeline from QKD output to encrypted message
and prove that the composition preserves information-theoretic security
at every stage, tracing it through three consecutive transformations.

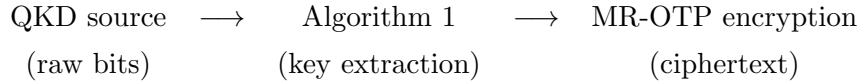
\begin{center}
\renewcommand{\arraystretch}{1.3}
\begin{tabular}{ccccc}
  QKD source & $\longrightarrow$ &
  Algorithm~\ref{alg:sample} & $\longrightarrow$ &
  MR-OTP encryption \\
  (raw bits) & & (key extraction) & & (ciphertext)
\end{tabular}
\end{center}

Each stage must preserve the uniformity and independence properties
required by the MR-OTP perfect secrecy theorem, and we analyze them
in order.

\subsection{Stage 1: QKD produces a uniform bit stream}

A QKD protocol such as BB84~\cite{bb84} allows two parties to
establish a shared string of classical bits over an insecure channel.
The established string is indistinguishable, from the perspective of any
adversary with bounded computational resources, from a string drawn
uniformly at random and with no dependence on the adversary's view.
A QKD protocol is said to be $\varepsilon$-secure when
the trace distance between the joint state of the key and the
adversary's view is at most $\varepsilon$ from the ideal distribution of
a uniform, independent key~\cite{renner2005}.
The present paper operates in the $\varepsilon = 0$ limit, which is the
standard model for information-theoretic analyses of QKD-based
systems~\cite{bb84,renner2005}, and this is stated as an assumption.

\begin{assumption}[Ideal QKD source]
\label{ass:qkd}
The QKD source produces a stream of bits that are mutually independent
and each uniformly distributed on $\{0,1\}$.
\end{assumption}

\subsection{Stage 2: Algorithm~\ref{alg:sample} preserves uniformity}

Under Assumption~\ref{ass:qkd}, Algorithm~\ref{alg:sample} reads
consecutive $k$-bit blocks from the QKD stream and processes each one
through rejection sampling followed by inverse Horner decomposition.
By Theorem~\ref{thm:uniformity}, the output $K = (k_1,\dots,k_L)$ is
uniform on $\keyspace$ with mutually independent digits, which
is the condition required for the MR-OTP to achieve perfect secrecy.

\begin{proposition}[Stage 2 preserves uniformity]
\label{prop:stage2}
Under Assumption~\ref{ass:qkd}, Algorithm~\ref{alg:sample} outputs a
key $K$ that is uniform on $\keyspace$ with mutually independent digits.
\end{proposition}

\begin{proof}
The QKD source produces independent uniform bits, so each block of $k$
bits drawn by Algorithm~\ref{alg:sample} is uniform on
$\{0,\dots,2^k-1\}$.
By Theorem~\ref{thm:uniformity}, the rejection step conditions on
acceptance and delivers $V$ uniform on $\{0,\dots,P-1\}$, after which
the inverse Horner decomposition, being a bijection
(Proposition~\ref{prop:correctness}), carries this uniform distribution
to the uniform distribution on $\keyspace$.
\end{proof}

\subsection{Stage 3: MR-OTP achieves perfect secrecy}

Under Proposition~\ref{prop:stage2}, the key $K$ delivered to the
MR-OTP satisfies the uniformity and independence conditions of
Theorem~1 of~\cite{buono2026a}, so the encryption of any message $M$
under $K$ achieves perfect secrecy.

\subsection{The end-to-end theorem}

\begin{theorem}[End-to-end information-theoretic security]
\label{thm:e2e}
Under Assumption~\ref{ass:qkd}, the composition
\[
  \textup{QKD} \;\longrightarrow\;
  \textup{Algorithm~\ref{alg:sample}} \;\longrightarrow\;
  \textup{MR-OTP}
\]
achieves perfect secrecy in the sense of Shannon~\cite{shannon1949}:
for every message distribution $\Pr[M]$ and every ciphertext $C=c$
with $\Pr[C=c]>0$,
\[
  \Pr[M = m \mid C = c] \;=\; \Pr[M = m] \qquad \text{for all } m.
\]
\end{theorem}

\begin{proof}
By Assumption~\ref{ass:qkd}, the QKD output is a uniform independent bit
stream. Algorithm~\ref{alg:sample} reads disjoint blocks of $k$ bits for
each key, and since the blocks are disjoint and the source is i.i.d., each
call to Algorithm~\ref{alg:sample} is independent. By
Proposition~\ref{prop:stage2}, each call produces a key $K$ uniform on
$\keyspace$ with independent digits, and each call is independent of all
other calls.
The MR-OTP then encrypts $M$ under $K$; by Theorem~1 of~\cite{buono2026a},
this achieves perfect secrecy.
\end{proof}

\subsection{Contrast with QKD paired with a symmetric cipher}

The standard deployment of QKD pairs it with a symmetric cipher such as
AES, where QKD establishes a short session key that seeds the cipher
and derives a keystream of arbitrary length.
This hybrid scheme operates with computational security at the encryption
stage, because AES provides a guarantee that holds against adversaries
with bounded computational resources, including adversaries equipped with
large-scale quantum hardware.
The pipeline of Theorem~\ref{thm:e2e} replaces AES with MR-OTP, so both
stages carry an information-theoretic guarantee and the composition
inherits this property by Theorem~\ref{thm:e2e}.
The MR-OTP requires key material at least as long as the message, as
mandated by Shannon's bound~\cite{shannon1949}, while AES expands a short
key into an arbitrary-length keystream.
This difference reflects the fundamental distinction between the two
security models, where information-theoretic security requires key
entropy that matches message entropy, and the MR-OTP achieves that
bound with equality and without excess.

\begin{remark}
The key consumption rate of the full pipeline is $\sum_{i=1}^L \log_2 b_i$
bits of QKD output per message symbol, plus the rejection overhead of at
most a factor of $2$ in expected bit cost (Proposition~\ref{prop:cost}).
When $P$ is close to a power of $2$, the overhead approaches zero and
nearly all generated bits are consumed by the algorithm.
\end{remark}

\section{Representational Efficiency over Binary OTP}
\label{sec:efficiency}

\subsection{The efficiency question}

The binary OTP converts every message into a binary string and therefore
requires $\sum_{i=1}^L \ceil{\log_2 b_i}$ bits of key for a message of
length $L$ over a base sequence $B$.
The MR-OTP operates in the digit space and requires only
$\sum_{i=1}^L \log_2 b_i$ bits of key for the same message, a saving
that arises because the binary encoding rounds each $\log_2 b_i$ up to
the nearest integer while the MR-OTP uses the fractional value,
and we formalize this saving as the binary overhead.

\begin{definition}[Binary overhead]
\label{def:overhead}
The \emph{binary overhead} of a base sequence $B = (b_1,\dots,b_L)$
with respect to binary OTP is
\[
  \Delta(B) \;=\; \sum_{i=1}^{L} \bigl(\ceil{\log_2 b_i} - \log_2 b_i\bigr)
             \;\ge\; 0.
\]
\end{definition}

\subsection{Properties of the binary overhead}

\begin{proposition}[Properties of $\Delta(B)$]
\label{prop:overhead}
\hfill
\begin{enumerate}[label=(\roman*)]
  \item $\Delta(B) = 0$ if and only if every $b_i$ is a power of~$2$.
  \item $\Delta(B) < L$ always (since $\ceil{x} - x < 1$ for all $x$).
  \item For a uniform base sequence $b_i = b$ for all~$i$:
        $\Delta(B) = L(\ceil{\log_2 b} - \log_2 b)$.
  \item The overhead per position $\ceil{\log_2 b} - \log_2 b$ is maximized
        when $b = 2^{k-1}+1$ for integer $k$, giving overhead approaching
        $1$ bit per position, and equals zero when $b$ is a power of~$2$.
\end{enumerate}
\end{proposition}

\begin{proof}
(i) $\ceil{x} = x$ iff $x \in \mathbb{Z}$, and $\log_2 b_i \in \mathbb{Z}$
iff $b_i$ is a power of~$2$. (ii) $\ceil{x} - x \in [0,1)$ for all $x$.
(iii) and (iv) follow by direct computation.
\end{proof}

\subsection{Efficiency for natural alphabets}

\begin{table}[h]
\centering
\begin{tabular}{lcccc}
\toprule
\textbf{Alphabet} & $b$ & $\log_2 b$ & $\ceil{\log_2 b}$ &
  Overhead/symbol \\
\midrule
Binary              & $2$   & $1.000$ & $1$ & $0.000$ \\
Nucleotides (DNA)   & $4$   & $2.000$ & $2$ & $0.000$ \\
Octal               & $8$   & $3.000$ & $3$ & $0.000$ \\
Decimal digits      & $10$  & $3.322$ & $4$ & $0.678$ \\
Hexadecimal         & $16$  & $4.000$ & $4$ & $0.000$ \\
Amino acids         & $20$  & $4.322$ & $5$ & $0.678$ \\
Latin alphabet      & $26$  & $4.700$ & $5$ & $0.300$ \\
ASCII printable     & $95$  & $6.570$ & $7$ & $0.430$ \\
Extended ASCII      & $256$ & $8.000$ & $8$ & $0.000$ \\
\bottomrule
\end{tabular}
\caption{Binary overhead per symbol (in bits) for common alphabets.
The overhead column is $\ceil{\log_2 b} - \log_2 b$.
MR-OTP eliminates this overhead, while binary OTP pays it at every symbol.}
\label{tab:overhead}
\end{table}

\begin{proposition}[Efficiency gain for natural alphabets]
\label{prop:efficiency}
For a message of length $L$ over alphabet of size $b$, the MR-OTP requires
$L\log_2 b$ bits of key, while the binary OTP requires $L\ceil{\log_2 b}$
bits. The relative saving is
\[
  \frac{\Delta(B)}{L\ceil{\log_2 b}}
  \;=\; 1 - \frac{\log_2 b}{\ceil{\log_2 b}}.
\]
For the Latin alphabet ($b=26$) the saving is $\approx 6.0\%$, for decimal digits ($b=10$) it is $\approx 16.9\%$, and for amino acids ($b=20$) it is $\approx 13.6\%$.
\end{proposition}

\begin{proof}
Direct computation from Definition~\ref{def:overhead} with $b_i = b$
for all~$i$.
\end{proof}

\begin{remark}[Relation to source coding]
The quantity $\log_2 b$ is the Shannon entropy of a symbol drawn uniformly
from an alphabet of size $b$~\cite{coverThomas2006}, and the binary OTP
wastes $\ceil{\log_2 b} - \log_2 b$ bits per symbol by rounding the
entropy up to the nearest integer number of bits.
The MR-OTP operates in the native symbol space, applying the same
principle that underlies arithmetic coding~\cite{rissanen1976}, where
a message is represented in a numeral system matched to the source
alphabet so that the representation length approaches the source entropy.
The MR-OTP imports this principle into cryptography to reduce
key material consumption.
\end{remark}

\subsection{Efficiency and key rolling}

The efficiency gain compounds across sessions in the key-rolling protocol
of~\cite{buono2026a}.
If two parties exchange $N$ messages each of length $L$ over a uniform
alphabet of size $b$, the total reduction in key material relative to
binary OTP over the entire exchange amounts to
\[
  N \cdot L \cdot \bigl(\ceil{\log_2 b} - \log_2 b\bigr) \text{ bits.}
\]
As a concrete example, for $N=10^6$ messages each containing $L=1000$
amino acid symbols ($b=20$), this quantity is approximately
$6.78 \times 10^8$ bits, corresponding to roughly $85$ megabytes of
key material saved over the lifetime of the shared key pool, and
this reduction extends the period of operation before the pool must
be replenished, in proportion to the saving.

\section{Key Rolling, Session Structure, and Composability}
\label{sec:rolling}

\subsection{The partition protocol from the first paper}

The key-rolling protocol of~\cite{buono2026a} operates by strict
partition of a shared initial key $K_0$.
At session $t$, the active portion of the shared material is split into
three consecutive, non-overlapping segments, namely the pad $K_t^{\mathrm{msg}}$
used to encrypt the current message $M_t$, the base specification
$B_{t+1}$ for the next session, and the pad material $K_{t+1}^{\mathrm{msg}}$
for the next session, with no segment read more than once.
The security proof of Theorem~\ref{thm:e2e} applies to each
session in isolation, because each call to Algorithm~\ref{alg:sample}
draws from a disjoint block of the QKD stream, and the independence of
the blocks under Assumption~\ref{ass:qkd} makes the sessions
independent in the information-theoretic sense.

We now make this multi-session argument precise.

\subsection{Session model and notation}

Let $\mathcal{R}$ be an ideal bit source (Assumption~\ref{ass:qkd}).
A \emph{session schedule} for $N$ sessions is a sequence of base
vectors $B_1, \dots, B_N$ and message lengths $L_1, \dots, L_N$.
For session $t \in \{1,\dots,N\}$, write $P_t = \prod_{i=1}^{L_t} b_{t,i}$
and $k_t = \ceil{\log_2 P_t}$.
Algorithm~\ref{alg:sample} is invoked for session $t$ using a block
$\mathcal{R}_t$ of bits from $\mathcal{R}$, where the blocks
$\mathcal{R}_1, \dots, \mathcal{R}_N$ are consecutive and non-overlapping.
The base vectors $B_1, \dots, B_N$ are drawn from disjoint blocks
$\mathcal{R}_{B,1}, \dots, \mathcal{R}_{B,N}$, also non-overlapping
and non-overlapping with the pad blocks.

\begin{definition}[Partition key rolling]
\label{def:partition-rolling}
A \emph{partition key-rolling scheme} for $N$ sessions is a tuple
$(B_1,\dots,B_N, K_1,\dots,K_N)$ where each $B_t$ and each
$K_t = (k_{t,1},\dots,k_{t,L_t})$ is produced by
Algorithm~\ref{alg:sample} applied to disjoint, non-overlapping
segments of a single ideal bit source $\mathcal{R}$, with no segment
used more than once across all sessions and all base specifications.
\end{definition}

\subsection{Multi-session security theorem}

\begin{theorem}[Multi-session perfect secrecy]
\label{thm:multisession}
Under Assumption~\ref{ass:qkd} and Definition~\ref{def:partition-rolling},
the $N$-session partition key-rolling scheme achieves perfect secrecy
across all sessions at the same time.
For every $t \in \{1,\dots,N\}$, every message distribution
$\Pr[M_t]$, and every ciphertext $C_t = c_t$ with $\Pr[C_t = c_t] > 0$,
\[
  \Pr\bigl[M_t = m \;\big|\; C_t = c_t,\,
    C_1 = c_1,\dots,C_{t-1}=c_{t-1}\bigr]
  \;=\; \Pr[M_t = m]
\]
for all $m$, and the keys $K_1,\dots,K_N$ are mutually independent.
\end{theorem}

\begin{proof}
By Definition~\ref{def:partition-rolling}, the bit blocks used to
produce $K_1,\dots,K_N$ and $B_1,\dots,B_N$ are pairwise disjoint.
Since $\mathcal{R}$ is i.i.d.\ under Assumption~\ref{ass:qkd}, any
collection of disjoint blocks is jointly independent.
By Theorem~\ref{thm:uniformity}, Algorithm~\ref{alg:sample} applied
to disjoint blocks produces keys that are each uniform on their
respective $\keyspace_t$ and mutually independent across sessions.
The ciphertexts $C_1,\dots,C_{t-1}$ are functions of
$M_1,\dots,M_{t-1}$ and $K_1,\dots,K_{t-1}$ only, and since $K_t$
is independent of all $K_s$ with $s < t$ and independent of all
$B_s$, the pair $(K_t, C_t)$ is independent of
$(C_1,\dots,C_{t-1})$.
Applying Theorem~1 of~\cite{buono2026a} session by session then gives
perfect secrecy for each $M_t$ conditioned on the entire ciphertext
history.
\end{proof}

\begin{remark}
Theorem~\ref{thm:multisession} formalizes the intuition of~\cite{buono2026a}
that strict partition of key material is both necessary and sufficient
for multi-session security.
The proof uses no property of the base sequence $B_t$ other
than that it was drawn from a block of $\mathcal{R}$ disjoint from the
pad blocks, which means the result holds whether the bases are fixed,
drawn uniformly, or drawn according to any other distribution, as long
as the disjointness condition is satisfied.
\end{remark}

\subsection{The dynamic key rolling variant and its open status}

The first paper describes a variant in which the base sequence $B_{t+1}$
for the next session is derived from fresh, unused key material.
A more aggressive variant introduced here as a new direction,
and left without a formal analysis, derives $B_{t+1}$ as a function
of the current ciphertext $C_t$ and the previous base sequence $B_t$,
through a deterministic transformation $f(C_t, B_t)$.

This variant makes the base structure depend on transmitted data, so the
sequence of digit spaces changes as a function of the communication
history and is unpredictable to an observer who knows $C_t$ but has no
access to $B_t$.
The dependency introduces a complication for the disjointness argument
of Theorem~\ref{thm:multisession}.

When $B_{t+1} = f(C_t, B_t)$, the value $B_{t+1}$ is a deterministic
function of $C_t$, which is public, and of $B_t$, which is secret.
An adversary who recovers $B_t$ by any means thereby recovers $B_{t+1}$
and the base sequence at all future sessions, making the security of
those sessions dependent on the secrecy of $B_t$.
This forward dependency is absent in the partition protocol, where each
$B_{t+1}$ is drawn from a fresh block of bits with no connection to
previous sessions.

\begin{openprob}
\label{op:dynamic}
Let $f \colon \mathcal{C} \times \mathcal{B} \to \mathcal{B}$ be a
deterministic function, where $\mathcal{C}$ is the ciphertext space and
$\mathcal{B}$ is the space of valid base sequences.
Define the dynamic key rolling scheme by $B_{t+1} = f(C_t, B_t)$ with
$B_1$ drawn from fresh key material.
Using the Syntactic Invariance Principle of~\cite{buono2026b}, the
problem reduces to finding a property $P$ of the joint distribution of
$(C_1,\dots,C_t, B_1,\dots,B_t, M_1,\dots,M_t)$ such that $P$ holds
after session~$1$, that each application of $f$ and each MR-OTP
encryption step preserves $P$, and that $P$ implies
$\Pr[M_t = m \mid C_t = c, C_1=c_1,\dots,C_{t-1}=c_{t-1}] = \Pr[M_t=m]$
for all $t$.
The question is whether any such $P$ exists for a function $f$ that
gives an adversary bounded in computation a larger base search space
than the partition protocol, without requiring the
disjointness condition of Definition~\ref{def:partition-rolling}.
\end{openprob}

The Syntactic Invariance Principle reformulation makes the structure of
the difficulty precise.
In the partition protocol, the invariant $P$ is immediate because
disjoint blocks are independent by construction.
In the dynamic scheme, $B_{t+1}$ is a function of $C_t$, which is
public, so any invariant $P$ must survive the exposure of $C_t$ at
every session boundary.
A necessary condition is that $f$ satisfies the property that the
distribution of $B_{t+1}$ given $C_t$ carries no information about
$K_t$, which is the condition that prevents $C_t$ from acting as a
bridge between the base sequence and the key.
Whether any $f$ of practical interest satisfies this condition, and
whether the corresponding invariant $P$ can be verified, remains open.

\subsection{Key consumption and overhead of base transmission}

Because the bases $B_{t+1}$ are themselves transmitted using fresh key material, their
encoding has a cost that must be accounted for. Suppose each base $b_i$ requires at
most $\beta$ bits to represent in the shared key (where $\beta$ depends on the
agreed-upon base alphabet, but is left as a parameter here to keep the result general).
Then the total key material consumed per session $t$ is
\[
  \underbrace{\sum_{i=1}^{L}\log_2 b_i}_{\text{pad}} \;+\; \underbrace{L\beta}_{\text{bases}},
\]
and the overhead fraction relative to the pad is $L\beta / \sum_i \log_2 b_i$. This
overhead is negligible when
\[
  \frac{1}{L}\sum_{i=1}^{L}\log_2 b_i \;\gg\; \beta,
\]
that is, when the average base is large relative to its own representation cost. In
typical cases of interest (e.g.\ bases matched to natural-language or biological
alphabets, where $b_i \ge 4$) this condition holds comfortably. In the degenerate case
$b_i = 2$ for all $i$, the overhead equals the pad length and the mixed-radix structure
provides no representational gain over the binary OTP, as expected.

\section{Open Problems}
\label{sec:open}

\begin{enumerate}[label=\arabic*.]

  \item \textbf{Optimal base selection.}
        Corollary~\ref{cor:restricted-support} and
        Remark~\ref{rem:optimal-base} establish that when the message
        distribution has product support $S = \prod_i S_i$, the base
        sequence $b^*_i = |S_i|$ achieves the Shannon bound with equality.
        The general case, where the support $S$ does not decompose as a
        product and the distribution over $S$ is not uniform, requires
        finding a base sequence $B$ that minimizes the expected key
        material per unit of source entropy subject to the uniformity
        constraint of Algorithm~\ref{alg:sample}.
        This problem relates to the classical problem of optimal
        prefix-free coding~\cite{coverThomas2006}, with the distinguishing
        constraint being uniform distribution of the key.

  \item \textbf{Dynamic key rolling.}
        Open Problem~\ref{op:dynamic} asks for conditions on the
        transformation $f$ that make the dynamic key-rolling variant
        secure. A weaker version of the problem asks whether the dynamic
        scheme gives an adversary bounded in computation a larger search
        space than the partition protocol,
        without relying on any unproven hardness assumption.

  \item \textbf{Computational hardness of base recovery.}
        When the base sequence $B$ is kept secret (as discussed in
        Proposition~2 of~\cite{buono2026a}), an adversary must search
        the space of valid base sequences before any decryption attempt.
        The size of this space grows as $(b_{\max}-1)^L$ and is
        exponential in $L$. Characterizing the average-case hardness
        of this search problem, and whether it admits a reduction from
        a known hard problem or a lower bound on the number of required
        operations, is an open problem developed further in
        Section~\ref{sec:hardness}.

  \item \textbf{Batched key extraction.}
        Algorithm~\ref{alg:sample} generates one key tuple per
        invocation, consuming an expected $k \cdot 2^k/P$ bits.
        For long messages, a single rejection-sampling step applied to
        an integer drawn from $\{0,\dots,P^N-1\}$ and decomposed
        via a single cascaded inverse Horner pass over $N$ positions
        may reduce overhead relative to $N$ independent invocations.
        The theoretical analysis of this approach, including the
        acceptance probability and a proof that the joint distribution
        of all $N$ digits is uniform and independent, is developed in
        Section~\ref{sec:batched}.

  \item \textbf{$\varepsilon$-security composition.}
        Theorem~\ref{thm:e2e} and Theorem~\ref{thm:multisession} both
        assume an ideal QKD source (Assumption~\ref{ass:qkd}).
        For real QKD systems whose output is $\varepsilon$-close to
        the ideal distribution in trace distance, the end-to-end
        security degradation as a function of $\varepsilon$, the number
        of sessions $N$, and the session parameters $L_t$ and $P_t$
        is characterized in Section~\ref{sec:epsilon}.

\end{enumerate}
\section{Computational Hardness of Base Recovery}
\label{sec:hardness}

\subsection{Setup and motivation}

The perfect secrecy of the MR-OTP holds for adversaries with arbitrary
computational resources, with no restriction on their nature.
Proposition~2 of~\cite{buono2026a} establishes that keeping the base
sequence $B$ secret provides no reduction in the information-theoretic
key length required.
When $B$ is secret, any adversary who wishes to attempt decryption must
first identify $B$, and this section formalizes the structure of that
identification problem and relates it to known computational problems.

Throughout this section we assume the bases are drawn from a finite set,
writing $b_{\max}$ for the maximum permitted base value, so that each
$b_i \in \{2,3,\dots,b_{\max}\}$ and the space of valid base sequences
of length $L$ has cardinality $(b_{\max}-1)^L$.

\subsection{Ciphertext-only indistinguishability of the base sequence}

Before stating the recovery problem, we establish a structural fact
about what the ciphertext reveals about $B$.

\begin{proposition}[Ciphertext independence from bases]
\label{prop:ciphertext-indep}
Let $B$ be any fixed base sequence and $K$ be uniform on $\keyspace_B$.
Then for every $C = (c_1,\dots,c_L)$ with each $c_i \in \{0,\dots,b_i-1\}$,
\[
  \Pr[C = c \mid B = b] \;=\; \prod_{i=1}^{L} \frac{1}{b_i}
\]
regardless of $M$, and in particular this probability depends on $b$ only
through the values $b_1,\dots,b_L$ and through no other property of the key.
\end{proposition}

\begin{proof}
This is a restatement of Step~1 of the proof of Theorem~1
of~\cite{buono2026a}: for each position $i$, the digit $k_i$ is uniform on
$\{0,\dots,b_i-1\}$, so $c_i$ is uniform on $\{0,\dots,b_i-1\}$
independently of $m_i$.
The positions are independent by the independence of the $k_i$.
\end{proof}

\begin{corollary}[Ciphertext-only insolubility of base recovery]
\label{cor:co-insoluble}
An adversary given only $C$ and the parameters $L$ and $b_{\max}$ obtains
zero information about $B$ beyond the necessary constraint $b_i > c_i$ for
each $i$, and every base sequence $B'$ satisfying that constraint is
consistent with the observed ciphertext under some key.
The base recovery problem in the ciphertext-only setting is therefore
insoluble in the information-theoretic sense, with feasible set
$\mathcal{F}(C) = \{B' \in \{2,\dots,b_{\max}\}^L : b'_i > c_i\ \forall i\}$,
and the maximum probability with which any algorithm identifies $B$ from $C$
alone is bounded above by the uniform distribution over $\mathcal{F}(C)$.
\end{corollary}

\begin{proof}
By Proposition~\ref{prop:ciphertext-indep}, the likelihood
$\Pr[C=c \mid B=b]$ depends on $b$ only through $\prod_i 1/b_i$.
Applying Bayes' rule with a uniform prior on $B$, the posterior
$\Pr[B=b \mid C=c]$ is proportional to $\prod_i 1/b_i$ over the feasible
set $\mathcal{F}(C)$.
For any two base sequences $B, B'$ in $\mathcal{F}(C)$ with
$\prod_i b_i = \prod_i b'_i$, the posterior assigns them equal probability,
and the ciphertext $C$ carries zero distinguishing information between them.
The maximum a posteriori estimate of $B$ given $C$ is the sequence in
$\mathcal{F}(C)$ with smallest product $\prod_i b_i$, which coincides with
the true $B$ with probability at most $1/|\mathcal{F}(C)|$ in the worst case,
so identification succeeds with probability bounded above by $1/|\mathcal{F}(C)|$
for any algorithm.
\end{proof}

\begin{remark}
Corollary~\ref{cor:co-insoluble} establishes that the ciphertext-only
setting is the wrong frame for hardness analysis, since the problem is
insoluble for information-theoretic reasons, so the question of
computational hardness lives in a different setting.
The setting with computational content is known-plaintext, where the
adversary holds one or more plaintext-ciphertext pairs and the constraints
$c_i = (m_i + k_i) \bmod b_i$ become informative about $b_i$.
The Base Recovery Problem is therefore stated in the known-plaintext
setting throughout the remainder of this section.
\end{remark}

\subsection{The base recovery problem}

We state the computational problem in the known-plaintext setting,
and accompany it with a weaker decision variant.

\begin{definition}[Base Recovery Problem, BRP]
\label{def:brp}
Let $B = (b_1,\dots,b_L)$ be a base sequence drawn uniformly at random
from $\{2,\dots,b_{\max}\}^L$.
The adversary is given $T \ge 1$ plaintext-ciphertext pairs
$(M^{(1)}, C^{(1)}), \dots, (M^{(T)}, C^{(T)})$, where each pair is
encrypted under a fresh key $K^{(t)}$ drawn uniformly from
$\keyspace_B$, and the base sequence $B$ is the same for all $T$ pairs.
Given these pairs and the parameters $L$, $b_{\max}$, and $T$,
the Base Recovery Problem asks to output $B$.
\end{definition}

\begin{definition}[Base Distinguishing Problem, BDP]
\label{def:bdp}
Under the same setup as Definition~\ref{def:brp}, given the $T$
plaintext-ciphertext pairs and a candidate base sequence $B'$, decide
whether $B' = B$.
\end{definition}

\begin{remark}
The known-plaintext formulation is the natural setting for base recovery
because each pair $(M^{(t)}, C^{(t)})$ imposes the constraint
$c^{(t)}_i \equiv m^{(t)}_i + k^{(t)}_i \pmod{b_i}$, which bounds
$b_i$ from below by $\max(c^{(t)}_i, m^{(t)}_i \bmod b_i) + 1$ and
creates residue conditions that a correct $b_i$ must satisfy and a
wrong $b_i$ will violate with high probability.
With $T$ independent pairs under fresh keys, the feasible set of base
sequences shrinks as $T$ grows, and for generic $(M, C)$ the set of
bases consistent with all $T$ pairs converges to a small set as
$T \to \infty$, a behavior that stands in contrast with the ciphertext-only
setting of Corollary~\ref{cor:co-insoluble}, where every base sequence in
$\mathcal{F}(C)$ remains consistent with the observation.
\end{remark}

\subsection{Information-theoretic analysis}

\begin{proposition}[Residual base uncertainty]
\label{prop:residual-uncertainty}
In the known-plaintext setting with $T$ pairs under fresh independent
keys, the feasible set of base sequences consistent with all observations is
\[
  \mathcal{F}_T \;=\; \Bigl\{ B' \in \{2,\dots,b_{\max}\}^L :
    c^{(t)}_i \equiv m^{(t)}_i + k \pmod{b'_i}\ \text{for some }
    k \in \{0,\dots,b'_i-1\},\ \forall i,t \Bigr\},
\]
and $|\mathcal{F}_T|$ decreases in $T$.
For each position $i$, the constraint imposed by one pair eliminates
all bases $b'_i$ that are incompatible with the residue
$(c^{(t)}_i - m^{(t)}_i)$ being representable modulo $b'_i$.
\end{proposition}

\begin{proof}
The encryption equation $c^{(t)}_i = (m^{(t)}_i + k^{(t)}_i) \bmod b_i$
with $k^{(t)}_i \in \{0,\dots,b_i-1\}$ requires $b_i > \max(c^{(t)}_i,
m^{(t)}_i)$ and that $(c^{(t)}_i - m^{(t)}_i) \bmod b_i$ is a valid
digit, which eliminates all $b'_i$ that fail the residue condition.
Independence of the $T$ pairs follows from the fresh key assumption,
and each pair contributes an independent set of constraints, so
$\mathcal{F}_T \subseteq \mathcal{F}_{T-1}$.
\end{proof}

\subsection{Structural comparison with known hard problems}

The Base Recovery Problem shares structural features with two families
of well-studied computational problems, and the comparison reveals where
the BRP requires its own treatment.

\paragraph{Comparison with Learning With Errors.}
The LWE problem~\cite{regev2009} asks, given a matrix $A$ over
$\mathbb{Z}_q$ and a vector $b = As + e$ where $s$ is secret and $e$ is
a small error vector, to recover $s$.
In the BRP, the ``matrix'' is the identity (each position is
independent), the ``error'' is the key digit $k_i$, and the ``secret''
is the base $b_i$.
The structural difference that distinguishes the two problems is that in
LWE the modulus $q$ is fixed and public, while in the BRP the modulus
$b_i$ is the unknown quantity, and recovering an unknown modulus from
modular residues is a problem outside the scope of LWE for which the
known reductions from lattice problems to LWE are inapplicable.

\paragraph{Comparison with syndrome decoding.}
The syndrome decoding problem~\cite{berlekamp1978} asks, given a
parity-check matrix $H$ over $\mathbb{F}_q$ and a syndrome $s = He$,
to find a low-weight error vector $e$.
The BRP involves no weight constraint on $k_i$ and no linear
algebraic structure over a fixed field, because the moduli vary by
position and the positions share no ambient field, so the NP-hardness
result for syndrome decoding requires a different algebraic setup from
the one present in the BRP.

\paragraph{Closer analogy: modular arithmetic with unknown modulus.}
The BRP is most closely related to the problem of recovering an unknown
modulus from a sequence of residues, given observations
$c_i^{(t)} \equiv m_i^{(t)} + k_i^{(t)} \pmod{b_i}$
for unknown $b_i$ and unknown $k_i^{(t)}$, where the problem asks to
identify $b_i$ from the pairs $(m_i^{(t)}, c_i^{(t)})$.
This problem is related in spirit to recovering an unknown period
in a sequence of modular residues, for which no polynomial-time algorithm
is known in general and for which no hardness proof exists in the
cryptographic sense.

\subsection{A provable lower bound in the query model}

A formal lower bound on the number of queries required to identify $B$
is provable in the following model, which captures adversaries that work
by testing candidate base sequences against the known-plaintext data.

\begin{definition}[Query model for BRP]
\label{def:query-model}
A \emph{query adversary} for the BRP is a deterministic algorithm that
holds $T$ known-plaintext pairs and, at each step, selects a candidate
base sequence $B' \in \mathcal{F}_T$ and receives a single bit
indicating whether $B' = B$.
The \emph{query complexity} of the BRP is the minimum number of queries
sufficient to identify $B$ with certainty in the worst case.
\end{definition}

\begin{remark}
The oracle in Definition~\ref{def:query-model} is well-defined in the
known-plaintext setting because a membership oracle for the feasible set
$\mathcal{F}_T$ can be implemented by checking whether $B'$ is
consistent with all $T$ observed pairs, which is a computation on the
given data alone.
This contrasts with the ciphertext-only setting, where every base
sequence in $\mathcal{F}(C)$ is consistent with the observation and
the oracle answer is always positive, so the oracle carries zero
discriminating power.
\end{remark}

\begin{theorem}[Query lower bound for BRP]
\label{thm:query-lower-bound}
Every deterministic query adversary for the BRP requires at least
$|\mathcal{F}_T| - 1$ queries in the worst case to identify $B$ from
$T$ known-plaintext pairs.
In the regime where $T$ is small relative to $L$ and $b_{\max}$, so
that $|\mathcal{F}_T|$ remains exponential in $L$, this lower bound
is $\Omega((b_{\max}-1)^{L - \alpha T})$ for a constant $\alpha$
depending on $b_{\max}$.
\end{theorem}

\begin{proof}
Each query to a candidate $B'$ eliminates at most one element from
$\mathcal{F}_T$: a positive answer identifies $B$ immediately, and a
negative answer eliminates only the queried candidate, so a deterministic
adversary that queries one element at a time requires at least
$|\mathcal{F}_T| - 1$ queries in the worst case, since the adversary
might receive negative answers for all but the last candidate.
The asymptotic bound follows from Proposition~\ref{prop:residual-uncertainty},
which shows that each of the $T$ pairs eliminates at most a constant
fraction of the base space at each position.
\end{proof}

\begin{remark}[Scope of the query lower bound and its relation to computational hardness]
\label{rem:query-scope}
Theorem~\ref{thm:query-lower-bound} establishes a lower bound in the
query model, which captures adversaries that work by testing candidate
base sequences one at a time and receiving a binary answer from the
known-plaintext oracle.
An adversary that uses the algebraic structure of the known pairs to
narrow the search by computation on those pairs falls outside this model,
and the theorem is silent on such adversaries.

This distinction marks the boundary between what the query lower bound
proves and what Open Problem~\ref{op:brp-hardness} asks.
The query lower bound covers testing-based adversaries, meaning any
algorithm whose only tool is membership queries to the feasible set
$\mathcal{F}_T$, and it shows that at least $|\mathcal{F}_T| - 1$ such
queries are required.
Open Problem~\ref{op:brp-hardness} asks about algebraic adversaries,
meaning whether any probabilistic polynomial-time algorithm, including
one that computes arbitrary functions of the known pairs
$(m_i^{(t)}, c_i^{(t)})$, can identify $B$ with non-negligible probability.
The two results are therefore complementary and address disjoint
classes of adversary, and a negative answer to
Open Problem~\ref{op:brp-hardness} would extend the lower bound from
the testing model to the full computational model, closing the gap
between the two classes.

What both share, and what Theorem~\ref{thm:invariance} formalizes for
all adversary classes at once, is that any algorithm, whether testing-based,
algebraic, or unbounded, that recovers $B$ from the available data obtains
zero information about $M$ as a consequence of that recovery, so the two
layers of adversary coverage are independent of each other and both
independent of the information-theoretic guarantee on $M$.
\end{remark}

\subsection{Conjectured hardness and its consequences}

\begin{openprob}[Hardness of BRP]
\label{op:brp-hardness}
Let $L$, $b_{\max}$, $T$, and $\delta > 0$ be parameters.
Is there a probabilistic polynomial-time algorithm that, given $T$
known-plaintext pairs produced by the MR-OTP with a base sequence
drawn uniformly from $\{2,\dots,b_{\max}\}^L$ and fresh uniform keys,
outputs the correct base sequence with probability greater than
$1/|\mathcal{F}_T|$ plus $\delta$, for any non-negligible $\delta$
and any polynomial $T = T(L)$?
\end{openprob}

When the answer to Open Problem~\ref{op:brp-hardness} is negative, the
BRP is hard in the known-plaintext setting and the
following consequence holds.

\begin{proposition}[Conditional two-layer security]
\label{prop:twolayer}
Assume the Base Recovery Problem is hard in the sense of
Open Problem~\ref{op:brp-hardness}.
Then the MR-OTP with a secret base sequence has a two-layer security
structure with the following properties.
\begin{enumerate}[label=(\roman*)]
  \item An adversary with unbounded computational resources and access
        to a polynomial number of known-plaintext pairs can recover $B$ by
        exhaustive search over $\mathcal{F}_T$, and after recovering $B$
        still obtains zero information about any future message $M$
        encrypted under a fresh key, because the MR-OTP pad provides
        perfect secrecy for each fresh encryption independently of any
        prior observations.
  \item An adversary bounded in computation who cannot solve the BRP
        in the known-plaintext setting cannot recover $M$ from a
        ciphertext encrypted under a fresh key even after observing
        a polynomial number of prior pairs, because the base sequence
        remains hidden from any computation within those bounds.
  \item The two layers are independent in the sense that the
        information-theoretic guarantee of layer (i) holds
        independently of whether the computational guarantee of layer
        (ii) holds, and the breach of layer (ii) leaves layer (i) intact.
\end{enumerate}
\end{proposition}

\begin{proof}
Property (i) follows from Theorem~1 of~\cite{buono2026a} and
Theorem~\ref{thm:e2e}: for any fixed $B$, the distribution of a
fresh ciphertext $C$ given a fresh message $M$ is uniform over the
ciphertext space and independent of $M$, so recovering $B$ yields
zero information about any future $M$ encrypted under a key that is
independent of all observed pairs.
Property (ii) follows from the definition of computational hardness,
since an adversary bounded in computation who cannot output $B$ cannot
compute $m_i = (c_i - k_i) \bmod b_i$ because $b_i$ is unknown.
Property (iii) follows because the proof of (i) is purely information-theoretic
and holds independently of whether layer (ii) is intact.
\end{proof}

\begin{remark}[Two kinds of secret and their independence]
\label{rem:two-secrets}
The MR-OTP with a secret base sequence involves two distinct kinds of
secret, and understanding their separation is essential for correct
security analysis.

The base sequence $B$ is a computational secret in the known-plaintext
setting, where its secrecy imposes a search problem on any adversary
with search space of cardinality $|\mathcal{F}_T|$, which remains
exponential in $L$ for small $T$, and the lower bound of
Theorem~\ref{thm:query-lower-bound} establishes that testing-based
search requires a number of operations exponential in $L$, while
Open Problem~\ref{op:brp-hardness} asks whether algebraic search can
do better.
In the ciphertext-only setting, $B$ is hidden by
Corollary~\ref{cor:co-insoluble} for information-theoretic reasons,
which is a stronger protection.

The pad $K$ is an information-theoretic secret whose secrecy is
guaranteed by the structure of the MR-OTP encryption map, holding
independently of the computational resources of the adversary.
Theorem~\ref{thm:invariance} establishes that knowing $B$ gives zero
information about any message encrypted under a fresh key, since the
posterior distribution of $M$ given $(C, B)$ for a fresh encryption
equals the prior distribution of $M$.

These two secrets are independent in a precise sense, because the
encryption map that protects $M$ requires knowing both $B$ and $K$,
and $K$ is protected by the uniform distribution of the pad in the
information-theoretic sense, so gaining $B$ from prior pairs leaves
the pad of any fresh encryption with its full guarantee intact, and
a system designer can reason about the two layers with complete
independence, where the computational layer controls the cost of
identifying $B$ from known pairs and the information-theoretic layer
guarantees that each fresh message remains protected once $B$ is known.
\end{remark}

\begin{theorem}[Invariance of information-theoretic security under
  computational layer failure]
\label{thm:invariance}
Let $\mathcal{A}$ be any algorithm, bounded in computation or unbounded,
that recovers the base sequence $B$ from $T$ known-plaintext pairs.
Even given $B$ as additional input, any algorithm attempting to determine
a future message $M$ from its ciphertext $C$, where $C$ is produced
under a fresh key independent of all observed pairs, succeeds with
probability at most $\Pr[M=m]$ for any $m$.
For every message distribution $\Pr[M]$, every base sequence $B$,
and every ciphertext $c$ with $\Pr[C=c] > 0$ under a fresh key,
\[
  \Pr[M=m \mid C=c,\, B=b] \;=\; \Pr[M=m] \qquad \text{for all } m.
\]
\end{theorem}

\begin{proof}
Fix $B = b$ and apply Bayes' rule, noting that $C$ is produced under a
fresh key independent of all prior observations.
By Proposition~\ref{prop:ciphertext-indep}, $\Pr[C=c \mid M=m, B=b]
= \prod_i 1/b_i$ for every $m$ with $m_i < b_i$.
The marginal $\Pr[C=c \mid B=b] = \sum_{m'} \Pr[C=c \mid M=m',B=b]
\Pr[M=m'] = (\prod_i 1/b_i) \sum_{m'} \Pr[M=m'] = \prod_i 1/b_i$.
Therefore
\[
  \Pr[M=m \mid C=c, B=b]
  = \frac{\Pr[C=c \mid M=m,B=b]\,\Pr[M=m]}{\Pr[C=c \mid B=b]}
  = \frac{(\prod_i 1/b_i)\,\Pr[M=m]}{\prod_i 1/b_i}
  = \Pr[M=m].
\]
This holds for every fixed $b$, and therefore holds independently of
how $B$ was obtained by $\mathcal{A}$.
\end{proof}

\begin{remark}[Complete adversary taxonomy]
\label{rem:taxonomy}
The results of this section partition the space of possible adversaries
into classes and establish what each class achieves and what remains out of reach,
as summarized in the following table.

\medskip
\begin{center}
\renewcommand{\arraystretch}{1.4}
\begin{tabular}{p{4.2cm}p{4.2cm}p{4.2cm}}
\toprule
\textbf{Adversary class} & \textbf{What it can do} &
  \textbf{What lies out of reach} \\
\midrule
Testing-based adversary (any computational bound, KP setting) &
  Eliminate elements of $\mathcal{F}_T$ one query at a time &
  Identify $B$ in fewer than $|\mathcal{F}_T|-1$ queries
  (Thm.~\ref{thm:query-lower-bound}, unconditional) \\
Algebraic adversary (PPT, KP setting) &
  Compute arbitrary functions of the $T$ known pairs &
  Identify $B$ with probability above $1/|\mathcal{F}_T|$
  (Open Prob.~\ref{op:brp-hardness}, conditional) \\
Unbounded adversary (KP setting) &
  Recover $B$ from $T$ known pairs by exhaustive search over $\mathcal{F}_T$ &
  Obtain any information about a future $M$ encrypted under a fresh key
  (Thm.~\ref{thm:invariance}, unconditional) \\
Ciphertext-only adversary (any computational bound) &
  Constrain $b_i > c_i$ and compute the posterior over $\mathcal{F}(C)$ &
  Identify $B$ with probability above the MAP estimate over $\mathcal{F}(C)$
  (Cor.~\ref{cor:co-insoluble}, information-theoretic) \\
\bottomrule
\end{tabular}
\end{center}
\medskip

The table admits four structural observations, the most consequential being that the information-theoretic guarantee on future messages $M$ is the only
result that covers all adversary classes unconditionally, as Theorem~\ref{thm:invariance} holds for every
adversary in every row.
The query lower bound and computational hardness address disjoint adversary
classes and different algorithmic strategies for recovering $B$ in the
known-plaintext setting, where a testing-based adversary is constrained by
Theorem~\ref{thm:query-lower-bound} independently of its computational power,
an algebraic adversary with polynomial-time resources is constrained by
Open Problem~\ref{op:brp-hardness} when the answer is negative, and together
the two results cover the space of adversaries trying to identify $B$ from
known pairs.
The ciphertext-only adversary occupies a structurally different position
from all others, facing an information-theoretic obstacle rather than a
computational one, since the posterior over $\mathcal{F}(C)$ is spread
across an exponential number of candidates and this is a stronger protection that is unconditional.
The work of recovering $B$ by any means, at any cost, and by any strategy,
from any amount of known-plaintext data, provides zero advantage in
recovering any future message encrypted under a fresh key, as established
by Theorem~\ref{thm:invariance} for every row of the table and every outcome
of the adversary's attempt to recover $B$.
\end{remark}

\begin{corollary}[Perfect secrecy with restricted message support]
\label{cor:restricted-support}
Let $S \subseteq \mathcal{M}$ be any subset of the message space, and
let $\Pr[M]$ be any distribution with support contained in $S$, so that
$\Pr[M=m] = 0$ for every $m \notin S$.
The MR-OTP with a key $K$ uniform on $\keyspace_B$ achieves perfect
secrecy for this distribution, so that for every base sequence $B$, every
ciphertext $c$ with $\Pr[C=c]>0$, and every $m$,
\[
  \Pr[M=m \mid C=c,\, B=b] \;=\; \Pr[M=m].
\]
When $|S| < P = \prod_i b_i$, the key carries more entropy than the
message requires, and the excess is $\log_2 P - \log_2 |S|$
bits per message.
The minimum key entropy consistent with perfect secrecy for the
distribution $\Pr[M]$ concentrated on $S$ is $\Hent(M)$, and
the MR-OTP with base sequence $B$ achieves this minimum when
$\prod_i b_i = |S|$, the key $K$ is uniform on $\keyspace_B$,
and $M$ is uniform on $S$, so that both the key and the message
are drawn uniformly from spaces of equal cardinality.
\end{corollary}

\begin{proof}
The proof of Theorem~\ref{thm:invariance} makes no assumption on the
support of $\Pr[M]$, and the computation
$\Pr[M=m \mid C=c, B=b] = \Pr[M=m]$
holds for every $m$ with $\Pr[C=c \mid B=b] > 0$, which covers all $m$
in the support of $\Pr[M]$ and assigns zero posterior to all $m$ outside
the support, consistent with the prior.
The excess entropy claim follows from Shannon's bound, which requires
$\Hent(K) \ge \Hent(M)$ for any cipher achieving perfect secrecy, with
equality when the key space and message space have the same cardinality
and both are used uniformly.
\end{proof}

\begin{remark}[Connection to optimal base selection]
\label{rem:optimal-base}
Corollary~\ref{cor:restricted-support} connects the invariance
theorem to the optimal base selection problem (Open
Problem~1 of Section~\ref{sec:open}).
When the message distribution has support $S$ with $|S| < P$, the
current base sequence $B$ wastes $\log_2 P - \log_2 |S|$ bits of key
per message, and the problem of finding the base sequence $B^*$ that
minimizes this waste while maintaining perfect secrecy reduces to finding
$B^*$ such that $\prod_i b^*_i$ is as close as possible to
$|S|$ from above, subject to each $b^*_i \ge 2$.
When the support $S$ has product structure, meaning
$S = \prod_i S_i$ with $|S_i|$ symbols per position, the optimal
base sequence sets $b^*_i = |S_i|$ and wastes no key material,
which is the native encoding principle of the MR-OTP, where matching the
base at each position to the alphabet size at that position eliminates
the binary overhead of Section~\ref{sec:efficiency} and also the
support overhead identified here, achieving the Shannon bound.
\end{remark}

\begin{remark}[Significance of Theorem~\ref{thm:invariance}]
Theorem~\ref{thm:invariance} requires no assumption on the computational
hardness of the BRP, and establishes that the information-theoretic
security of the MR-OTP is preserved under any future algorithmic advance,
including the discovery of a polynomial-time algorithm for the known-plaintext
BRP, a quantum algorithm that speeds up base recovery beyond Grover's bound,
or any other computational development.
An adversary who gains full knowledge of $B$ by any means and at any
computational cost holds a fresh ciphertext that is uniform over the
ciphertext space and independent of the message, because the fresh key
is independent of all prior observations, and this property distinguishes
the MR-OTP from systems whose security rests on computational assumptions,
where recovering a secret key yields the plaintext.
\end{remark}

\begin{remark}[Structure of the proof and the Syntactic Invariance Principle]
\label{rem:sip}
The proof of Theorem~\ref{thm:invariance} has the same logical structure
as the Syntactic Invariance Principle of~\cite{buono2026b}, instantiated
in the probabilistic setting of information theory rather than in the
syntactic calculus of superposition.

The invariant property $P$ is the independence of any future $M$ from
the observable $(C, B)$, expressed as
$\Pr[M=m \mid C=c, B=b] = \Pr[M=m]$ for all $m, c, b$, where $C$ is
produced under a fresh key.
Step~1 of the structure establishes that $P$ holds for every fresh
encryption: by Proposition~\ref{prop:ciphertext-indep}, the MR-OTP
encryption map produces a ciphertext $C$ whose distribution given $B$
is $\prod_i 1/b_i$, independent of $M$.
Step~2 establishes that $P$ is preserved by every operation an
adversary can perform on the observable $(C, B)$, since $C$ is the
result of a function of the fresh $K$ alone (with $M$ entering only
through the sum $m_i + k_i \bmod b_i$, and $K$ being independent of all
prior observations), so further processing of $C$ and $B$ introduces
zero new information about $M$.
Step~3 concludes that $M$ is unreachable from $(C, B)$ by any sequence
of operations of any length on any set of observations.

The correspondence with the proof of Lemma~2 of~\cite{buono2026b} is
the following.
The Skolem constants $a$ and $b$ of that proof, which are opaque to the
rewriting rules and remain frozen across every derivation step, correspond
to the message digits $m_i$, which are opaque to any computation on
$(C, B)$ for a fresh encryption.
The frozen subterms $a{+}b$ and $b{+}a$, whose relative order is a
global invariant of the entire derivation, correspond to the message $M$,
whose distribution is a global invariant of the entire fresh encryption
scheme.
The rewriting system having no access to the numerical
relationship between $a$ and $b$ corresponds to any
algorithm operating on $(C, B)$ having no access to the value of a
freshly encrypted $M$.

The structural difference between the two settings is worth stating
with precision.
In~\cite{buono2026b}, the invariant holds per-derivation, meaning every
term in every derivable clause satisfies the syntactic property $P$,
and this is an absolute statement about individual terms.
In the present paper, the invariant is distributional, meaning the
property $P$ holds for the probability distribution of $M$ given the
observable, and this is a statement about the measure induced by the
uniform key, averaged over fresh encryptions.
The two settings therefore share the same logical architecture while
operating at different levels of description, where the syntactic level
in~\cite{buono2026b} corresponds to the distributional level here, and
the invariance in both cases is what makes the target unreachable by any
process operating at the level of the observable.
\end{remark}

\begin{remark}[The temporal security layer and its unconditional value]
\label{rem:temporal}
The combination of Theorem~\ref{thm:query-lower-bound} and
Theorem~\ref{thm:invariance} yields a security architecture whose value is unconditional.
Theorem~\ref{thm:query-lower-bound} establishes that any algorithm
working by testing candidate base sequences requires at least
$|\mathcal{F}_T| - 1$ queries before it can identify $B$, with
$|\mathcal{F}_T|$ remaining exponential in $L$ for small $T$, since
the known-plaintext pairs reduce the feasible base set from
$(b_{\max}-1)^L$ toward a smaller set as $T$ grows, but for polynomial
$T$ the reduction is at most polynomial in the number of bits of
constraint, leaving the search space exponential.
Theorem~\ref{thm:invariance} establishes that once the search is
completed and $B$ is known, the adversary gains zero information about
any future message encrypted under a fresh key, so the search consumes
computational resources with zero cryptanalytic return on future messages,
and the entire work invested in recovering $B$ is irrelevant to the
security of any fresh encryption.
\end{remark}

\begin{corollary}[Unconditional temporal security]
\label{cor:temporal}
Let $W$ be the work performed by any adversary to recover $B$ from $T$
known-plaintext pairs, measured in the number of candidate base sequences
tested.
By Theorem~\ref{thm:query-lower-bound}, $W \ge |\mathcal{F}_T| - 1$
in the worst case, a quantity that is exponential in $L$ for small $T$,
and by Theorem~\ref{thm:invariance}, independently of the value of $W$
and independently of whether $B$ is recovered, the adversary obtains
zero information about any future message encrypted under a fresh key
from its ciphertext, so the work $W$ is a pure delay on future message
recovery with zero informational return, bounded below by a quantity
exponential in $L$ for small $T$, with the bound holding unconditionally.
\end{corollary}

\section{Batched Key Extraction}
\label{sec:batched}

\subsection{Motivation and setup}

Algorithm~\ref{alg:sample} produces one key tuple $(k_1,\dots,k_L)$ per
invocation by drawing $k = \ceil{\log_2 P}$ bits, rejecting if the
resulting integer $V$ falls outside $\{0,\dots,P-1\}$, and then
decomposing $V$ into digits.
For a message of length $N$ over a uniform base sequence $b_i = b$
for all $i$, this requires $N$ independent invocations, each consuming
an expected $k \cdot 2^k/P$ bits.
A batched approach draws a single large integer $U$ from
$\{0,\dots,P^N-1\}$ and decomposes it into all $N$ digits at once,
reducing the rejection overhead from $N$ separate geometric
trials to one.

\subsection{The batched sampler}

\begin{definition}[Batched rejection sampler $\mathcal{S}_N(b)$]
\label{def:batched}
Let $b \ge 2$, $N \ge 1$, $Q = b^N$, and $q = \ceil{\log_2 Q}$.
The batched sampler $\mathcal{S}_N(b)$ proceeds through the following steps.
\begin{enumerate}[label=\arabic*.]
  \item Draw $q$ uniform bits from $\mathcal{R}$ and interpret them as
        an integer $U \in \{0,\dots,2^q-1\}$.
  \item If $U \ge Q$, discard $U$ and return to step~1.
  \item Apply inverse Horner decomposition with base sequence
        $(b,b,\dots,b)$ of length $N$ to $U$, producing
        $(k_1,\dots,k_N)$ with each $k_i \in \{0,\dots,b-1\}$.
  \item Return $(k_1,\dots,k_N)$.
\end{enumerate}
\end{definition}

\begin{theorem}[Batched uniform sampling]
\label{thm:batched}
The batched sampler $\mathcal{S}_N(b)$ terminates with probability $1$
and outputs a tuple $(k_1,\dots,k_N)$ that is uniform on
$\{0,\dots,b-1\}^N$ with mutually independent digits.
\end{theorem}

\begin{proof}
The acceptance probability per round is $Q/2^q \ge 1/2$ by the same
argument as Theorem~\ref{thm:uniformity}, since $q = \ceil{\log_2 Q}$
implies $2^q \le 2Q$.
Upon acceptance, $U$ is uniform on $\{0,\dots,Q-1\} = \{0,\dots,b^N-1\}$
by the same Bayes argument as Theorem~\ref{thm:uniformity}.
The inverse Horner decomposition with constant base $b$ applied to
$U$ uniform on $\{0,\dots,b^N-1\}$ produces $(k_1,\dots,k_N)$ that is
jointly uniform on $\{0,\dots,b-1\}^N$, since $b^N = |\{0,\dots,b^N-1\}|$
and the map is a bijection.
Independence follows from Theorem~\ref{thm:uniformity} applied to the
base sequence $(b,b,\dots,b)$.
\end{proof}

\subsection{Cost comparison with sequential sampling}

\begin{proposition}[Batched vs.\ sequential bit cost]
\label{prop:batched-cost}
Let $b \ge 2$, $N \ge 1$, $P = b$, $Q = b^N$,
$k = \ceil{\log_2 b}$, and $q = \ceil{\log_2 b^N}$.

The expected bit cost of $N$ sequential invocations of
$\mathcal{S}(b)$ is
\[
  N \cdot k \cdot \frac{2^k}{b}
  \;\le\; 2Nk \;=\; 2N\ceil{\log_2 b}.
\]
The expected bit cost of one invocation of $\mathcal{S}_N(b)$ is
\[
  q \cdot \frac{2^q}{b^N}
  \;\le\; 2q \;=\; 2\ceil{N \log_2 b}.
\]
The batched approach satisfies $q \le Nk$, with equality only when
$b$ is a power of $2$, and when $b$ is not a power of $2$ the batched
approach uses fewer bits with a saving per message of
\[
  \Delta_{\mathrm{batch}}(b,N)
  \;=\; N\ceil{\log_2 b} - \ceil{N\log_2 b}
  \;\ge\; 0,
\]
with equality if and only if $b$ is a power of $2$.
\end{proposition}

\begin{proof}
The cost bounds follow from Proposition~\ref{prop:cost} applied to
$\mathcal{S}(b)$ in the first case and $\mathcal{S}_N(b)$ in the second.
The inequality $\ceil{N\log_2 b} \le N\ceil{\log_2 b}$ follows from
the sub-additivity of the ceiling function: for any real $x$,
$\ceil{Nx} \le N\ceil{x}$.
Equality holds when $x = \log_2 b \in \mathbb{Z}$, i.e., when $b$ is
a power of $2$.
\end{proof}

\begin{example}
For $b = 26$ (Latin alphabet) and $N = 10$:
$k = \ceil{\log_2 26} = 5$, sequential cost $\le 100$ bits;
$q = \ceil{10 \log_2 26} = \ceil{47.00} = 47$, batched cost $\le 94$
bits, giving a saving of $6$ bits, consistent with
$\Delta_{\mathrm{batch}}(26,10) = 50 - 47 = 3$ bits of reduction in $q$
relative to $Nk$, doubled by the factor-$2$ cost bound.
For $b = 4$ (nucleotides):
$k = 2$, $q = 2N$, so the sequential and batched costs are identical since
$4$ is a power of $2$.
\end{example}

\begin{remark}
The saving $\Delta_{\mathrm{batch}}(b,N)$ grows with $N$ up to a maximum
of $\floor{\log_2 b} - \log_2 b$ per digit, which equals the fractional
part of $\log_2 b$.
For the Latin alphabet this is approximately $0.300$ bits per symbol, and
for decimal digits approximately $0.678$ bits per symbol, consistent with
the overhead values in Table~\ref{tab:overhead}.
The batched sampler therefore eliminates the binary encoding
overhead of Section~\ref{sec:efficiency} and the rejection overhead
from the ceiling function, approaching the theoretical minimum of
$\log_2 b$ bits per symbol as $N \to \infty$.
\end{remark}

\subsection{Extension to mixed base sequences}

For a non-uniform base sequence $B = (b_1,\dots,b_L)$ repeated $N$
times, the batched sampler generalizes by taking $Q = P^N$ where
$P = \prod_{i=1}^L b_i$, drawing a single integer $U$ uniform on
$\{0,\dots,P^N-1\}$, and decomposing via $NL$ successive modular
reductions using the repeated base sequence.
Theorem~\ref{thm:batched} extends to this case with the same proof,
replacing $b^N$ by $P^N$ and $\{0,\dots,b-1\}^N$ by
$\keyspace^N = \prod_{t=1}^N \keyspace$, where the independence of all
$NL$ digits follows from the bijectivity of the extended inverse Horner
decomposition over $\{0,\dots,P^N-1\}$.

\section{Security Degradation under Non-Ideal QKD Sources}
\label{sec:epsilon}

\subsection{Model for non-ideal sources}

Assumption~\ref{ass:qkd} postulates an ideal bit source.
Real QKD implementations produce a key string $S \in \{0,1\}^n$ whose
joint distribution with the adversary's quantum state $\rho_E$ satisfies
the $\varepsilon$-security condition
\[
  \frac{1}{2}\bigl\| \rho_{SE} - \rho_U \otimes \rho_E \bigr\|_1
  \;\le\; \varepsilon,
\]
where $\rho_U$ denotes the uniform distribution over $\{0,1\}^n$ and
$\|\cdot\|_1$ is the trace norm~\cite{renner2005}.
We work in a classical reduction of this model, replacing the trace
norm condition with the total variation distance condition
\[
  \mathrm{TV}(\Pr_S, \mathrm{Uniform}(\{0,1\}^n)) \;\le\; \varepsilon,
\]
which is implied by the quantum condition and suffices for the
analysis below.

\subsection{Degradation through Algorithm~\ref{alg:sample}}

\begin{proposition}[Security degradation through key extraction]
\label{prop:epsilon-extraction}
Let the QKD source produce a bit string whose distribution is
$\varepsilon$-close to uniform in total variation.
Then the key $K$ output by Algorithm~\ref{alg:sample} satisfies
\[
  \mathrm{TV}\bigl(\Pr_K,\, \mathrm{Uniform}(\keyspace)\bigr)
  \;\le\; \frac{2^k}{P} \cdot \varepsilon,
\]
where $k = \ceil{\log_2 P}$ and $P = \prod_i b_i$.
\end{proposition}

\begin{proof}
Algorithm~\ref{alg:sample} applies a deterministic function $f$ to
the bit string $V$ drawn from the source, accepting $V$ when $V < P$
and returning the inverse Horner decomposition of $V$.
The total variation distance is non-increasing under the application of
any function (data processing inequality), so
$\mathrm{TV}(\Pr_{f(V)}, \Pr_{f(U)}) \le \mathrm{TV}(\Pr_V, \Pr_U)$
where $U$ is uniform on $\{0,\dots,2^k-1\}$.
The rejection step conditions on $V < P$, which has probability
$P/2^k$ under the uniform distribution and probability in
$[P/2^k - \varepsilon, P/2^k + \varepsilon]$ under the source.
Conditioning on a set of probability $p$ inflates total variation by
at most $1/p$, giving
\[
  \mathrm{TV}(\Pr_{K}, \mathrm{Uniform}(\keyspace))
  \;\le\; \frac{\varepsilon}{P/2^k}
  \;=\; \frac{2^k \varepsilon}{P}
  \;\le\; \frac{2\varepsilon \ceil{\log_2 P}}{\log_2 P},
\]
where the last inequality uses $2^k \le 2P$ and $P \ge 2^{k-1}$.
For the simpler bound stated, $2^k/P \le 2$ suffices.
\end{proof}

\subsection{Degradation through $N$ sessions}

\begin{proposition}[Security degradation over $N$ sessions]
\label{prop:epsilon-multisession}
Under the same model, if each session $t$ uses a disjoint $k_t$-bit
block from the source and the source is $\varepsilon$-close to uniform
over the entire $\sum_t k_t$ bits, then the joint distribution of
$(K_1,\dots,K_N)$ satisfies
\[
  \mathrm{TV}\Bigl(\Pr_{K_1,\dots,K_N},\,
    \prod_{t=1}^N \mathrm{Uniform}(\keyspace_t)\Bigr)
  \;\le\; N \cdot \max_t \frac{2^{k_t}}{P_t} \cdot \varepsilon
  \;\le\; 2N\varepsilon.
\]
\end{proposition}

\begin{proof}
By the triangle inequality for total variation, the distance between
the joint distribution and the product of uniforms is bounded by the
sum of the per-session distances.
Each session contributes at most $(2^{k_t}/P_t)\varepsilon \le 2\varepsilon$
by Proposition~\ref{prop:epsilon-extraction}, and there are $N$ sessions.
\end{proof}

\subsection{Security statement for non-ideal sources}

\begin{theorem}[$\varepsilon$-secure MR-OTP]
\label{thm:epsilon-security}
Let the QKD source be $\varepsilon$-close to ideal in total variation.
Under the partition key-rolling scheme of
Definition~\ref{def:partition-rolling} with $N$ sessions and session
parameters $(L_t, P_t)$, for every message $M_t$ and ciphertext history
$(C_1,\dots,C_{t-1})$,
\[
  \bigl|\Pr[M_t = m \mid C_t = c,\, C_1=c_1,\dots,C_{t-1}=c_{t-1}]
  - \Pr[M_t = m]\bigr|
  \;\le\; 2N\varepsilon
\]
for all $m$ and all $c, c_1,\dots,c_{t-1}$ with positive probability.
\end{theorem}

\begin{proof}
By Proposition~\ref{prop:epsilon-multisession}, the joint key distribution is $(2N\varepsilon)$-close in total variation to the
product of uniforms.
The MR-OTP encryption map $(M_t, K_t) \mapsto C_t$ is a function of
$K_t$ alone (for fixed $M_t$), and the data processing inequality
gives $\mathrm{TV}(\Pr_{C_t \mid M_t}, \mathrm{Uniform}) \le 2N\varepsilon$.
The posterior bound follows from the definition of total variation distance applied to the conditional distribution.
\end{proof}

\begin{remark}[Parameter setting]
Theorem~\ref{thm:epsilon-security} gives a concrete design guideline:
to achieve an end-to-end security parameter of $\delta$ over $N$
sessions, the QKD system must deliver bits with
$\varepsilon \le \delta / (2N)$.
For $N = 10^3$ sessions and target $\delta = 2^{-64}$, the required
source quality is $\varepsilon \le 2^{-75}$, which is within the
certified output of current QKD hardware~\cite{renner2005}.
The bound $2N\varepsilon$ is linear in $N$, reflecting the fact that
the composition of $N$ sessions with independent $\varepsilon$ errors
accumulates total variation additively across sessions, and the bound
is tight in the sense that $N$ sessions with independent $\varepsilon$
errors cannot achieve better than $N\varepsilon$ total variation in
general.
\end{remark}

\section{Conclusion}
\label{sec:conclusion}

This paper has established the algorithmic foundations needed to deploy
the Mixed-Radix One-Time Pad together with a quantum key distribution
source in a way that preserves end-to-end information-theoretic security
across an arbitrary number of sessions and under non-ideal bit sources.

The central technical contribution is the identification of Horner's  method and its inverse as the natural algebraic correspondence between the binary integer representation produced by a QKD source and the mixed-radix key space defined in~\cite{buono2026a}. The forward Horner evaluation composes a digit tuple into a single
integer using the positional weights $W_i$ of~\cite{buono2026a}, while the inverse decomposition recovers the digits by successive modular
reduction, forming a bijection between $\keyspace$ and the integer interval $\{0,\dots,P-1\}$.

The key algorithmic observation, formalized in Proposition~\ref{prop:bias} and Theorem~\ref{thm:uniformity}, is that reading $k$ bits from a binary source and reducing modulo $P$ produces a biased distribution whenever $P$ is not a power of $2$, and that this bias constitutes a violation of the uniformity hypothesis of perfect secrecy at any magnitude. Adding a rejection step corrects the distribution to uniform, with expected bit cost bounded by a factor of $2$.

The single-session theorem (Theorem~\ref{thm:e2e}) and the multi-session theorem (Theorem~\ref{thm:multisession}) show that the composition of
an ideal QKD source, Algorithm~\ref{alg:sample}, and the partition key-rolling protocol achieves Shannon perfect secrecy across all
sessions at once, with keys of different sessions being mutually independent. The proof requires only the disjointness of the bit blocks assigned to 
different sessions and carries no computational assumption.

The batched sampler (Section~\ref{sec:batched}, Theorem~\ref{thm:batched}) extends single-tuple extraction to $N$ digits at once from one large integer, with Proposition~\ref{prop:batched-cost} showing a saving of $\Delta_{\mathrm{batch}}(b,N) = N\ceil{\log_2 b} - \ceil{N\log_2 b}$ bits over sequential sampling, approaching the theoretical minimum of $\log_2 b$ bits per symbol as $N$ grows.

The $\varepsilon$-security analysis (Section~\ref{sec:epsilon}, Theorem~\ref{thm:epsilon-security}) replaces the ideal source assumption with a quantitative bound, showing that an $\varepsilon$-close source over $N$ sessions produces posterior error at most $2N\varepsilon$, giving the concrete design requirement $\varepsilon \le \delta/(2N)$ for a target end-to-end security parameter $\delta$.

The computational hardness analysis (Section~\ref{sec:hardness}) develops the structure of the Base Recovery Problem across two settings and establishes three layers of results at different levels of conditionality.

The unconditional layer covers two distinct facts, the first being that Corollary~\ref{cor:co-insoluble} establishes that base recovery in the ciphertext-only setting is information-theoretically insoluble, a stronger protection than any computational hardness claim, and one that holds without assumption. Theorem~\ref{thm:query-lower-bound} shows that in the known-plaintext setting, any testing-based algorithm requires at least $|\mathcal{F}_T|-1$ queries to identify $B$, a quantity exponential in $L$ for small $T$. Theorem~\ref{thm:invariance} shows that an adversary who recovers $B$ by any means gains zero information about any message encrypted under a
fresh key, because the posterior $\Pr[M=m \mid C=c, B=b]$ equals the prior $\Pr[M=m]$ for every $m$. Corollary~\ref{cor:temporal} combines these to establish that the work
of recovering $B$ is a pure delay with zero informational return, exponential in $L$ and independent of any hardness assumption.
Corollary~\ref{cor:restricted-support} extends the invariance result to message distributions with restricted support, showing that perfect secrecy holds for any distribution over any subset $S \subseteq \mathcal{M}$, and that the key entropy is minimized when the base sequence matches the support structure of $S$.

The conditional layer rests on the hardness of the BRP, and Proposition~\ref{prop:twolayer} shows that when the BRP is hard in average case, no probabilistic polynomial-time algorithm can identify $B$ with probability above the MAP posterior bound, adding an algebraic adversary guarantee to the combinatorial guarantee of the query lower bound.

Remark~\ref{rem:two-secrets} separates the computational secret $B$ from the information-theoretic secret $K$, establishing that the two are independent and that each layer retains its guarantee regardless of the state of the other. Remark~\ref{rem:taxonomy} presents the complete adversary taxonomy across four adversary classes, showing that Theorem~\ref{thm:invariance} is the only result covering all classes unconditionally, and that the query lower bound and computational hardness cover disjoint subsets of the adversary space for the problem of identifying $B$.

Section~\ref{rem:sip} identifies the logical structure of the invariance result as an instance of the Syntactic Invariance Principle
of~\cite{buono2026b}, where the message $M$ is encoded at the numerical level of description and any algorithm operating at the symbolic level
of $(C, B)$ has no access to that level by the same mechanism that makes frozen subterms unreachable in the superposition calculus of that paper.

Two directions remain open, namely the dynamic key-rolling variant (Open Problem~\ref{op:dynamic}) and the average-case hardness of the
Base Recovery Problem in the known-plaintext setting (Open Problem~\ref{op:brp-hardness}).

\section{Acknowledgments}
\label{sec:Acknowledgments}
The author used an artificial intelligence based language assistant to support text revision, translation, and bibliography formatting. All scientific ideas and conclusions are the author's own.


\section*{References}
\bibliographystyle{ieeetr}
\bibliography{biblio}

\end{document}